\documentclass[11pt]{article}

\usepackage{color}
\usepackage{xspace,epsfig} 
\usepackage{xspace,url}
\usepackage{fullpage}
\usepackage{graphicx}
\usepackage{amsmath, amssymb}
\usepackage{amsthm}
\usepackage{latexsym}
\usepackage{amsfonts}
\usepackage{hyperref}

\newtheorem{theorem}{Theorem}[section]
\newtheorem{lemma}[theorem]{Lemma}

\newtheorem{proposition}[theorem]{Proposition}

\newtheorem{corollary}[theorem]{Corollary}

\newenvironment{remark}{\smallskip\noindent {\em Remark:}}{\smallskip}

\DeclareSymbolFont{AMSb}{U}{msb}{m}{n}
\DeclareMathSymbol{\N}{\mathbin}{AMSb}{"4E}
\DeclareMathSymbol{\Z}{\mathbin}{AMSb}{"5A}
\DeclareMathSymbol{\R}{\mathbin}{AMSb}{"52}
\DeclareMathSymbol{\Q}{\mathbin}{AMSb}{"51}
\DeclareMathSymbol{\erert}{\mathbin}{AMSb}{"50}
\DeclareMathSymbol{\I}{\mathbin}{AMSb}{"49}
\DeclareMathSymbol{\C}{\mathbin}{AMSb}{"43}

\newcommand{\cO}{\mathcal{O}}
\newcommand{\CCC}{\mathcal{C}}
\newcommand{\BBB}{\mathcal{B}}
\newcommand{\bandw}{bandwidth\xspace}

\newcommand{\BANDW}{Bandwidth\xspace}
\newcommand{\smallcomp}[1]{\mathsf{sc}(#1)}
\newcommand{\smallcompvert}[1]{V(\smallcomp{#1})}
\newcommand{\largecomp}[1]{\mathsf{lc}(#1)}
\newcommand{\largecompvert}[1]{V(\largecomp{#1})}

\definecolor{officegreen}{rgb}{0.0, 0.5, 0.0}

\title{An Exponential Time 2-Approximation Algorithm for Bandwidth \thanks{A preliminary version of this paper appeared in the proceedings of IWPEC 2009 \cite{FurerGK09}.}}
\author{
 Martin F{\"u}rer\thanks{%
  Computer Science and Engineering, Pennsylvania State University, University Park, Pennsylvania, USA. 
  Email: \texttt{furer@cse.psu.edu}}
 \and
 Serge Gaspers\thanks{%
  Institute of Information Systems, Vienna University of Technology, Vienna, Austria.
  Email: \texttt{gaspers@kr.tuwien.ac.at}}
 \and
 Shiva Prasad Kasiviswanathan\thanks{%
  IBM T.\ J.\ Watson Research Center, Yorktown Heights, New York, USA.
  Email: \texttt{kasivisw@gmail.com}}
}

\begin{document}
\date{}
\maketitle

\begin{abstract}
The \bandw of a graph $G$ on $n$ vertices is the minimum $b$ such that the vertices of $G$ can be labeled from $1$ to $n$ such that the labels of every pair of adjacent vertices differ by at most $b$.

In this paper, we present a $2$-approximation algorithm for the \BANDW problem that takes worst-case $\cO(1.9797^n)$ $= \cO(3^{0.6217 n})$ time and uses polynomial space. This improves both the previous best $2$- and $3$-approximation algorithms of Cygan {\em et al.} which have $\cO^*(3^n)$ and $\cO^*(2^n)$ worst-case running time bounds, respectively.
Our algorithm is based on constructing bucket decompositions of the input graph. A bucket decomposition partitions the vertex set of a graph into ordered sets (called {\em buckets}) of (almost) equal sizes such that all edges are either incident to vertices in the same bucket or to vertices in two consecutive buckets. The idea is to find the smallest bucket size for which there exists a bucket decomposition. The algorithm uses a divide-and-conquer strategy along with dynamic programming to achieve the improved time bound.

\medskip\noindent
\emph{Keywords}: exponential time algorithm; approximation algorithm; graph bandwidth; bucket decomposition
\end{abstract}

\section{Introduction}
\label{sec:intro}
The \emph{\bandw{}} of a graph $G$ on $n$ vertices is the minimum integer $b$ such that the vertices of $G$ can be labeled from $1$ to $n$ such that the labels of every pair of adjacent vertices differ by at most $b$.
The \BANDW problem has as input a graph $G$ and an integer $b$ and the question is whether $G$ has \bandw at most $b$. The problem is a special case of Subgraph Isomorphism, as it can be formulated as follows: Is $G$ isomorphic to a subgraph of $P_n^b$? Here, $P_n^b$ denotes the graph obtained from $P_n$ (the path on $n$ vertices) by adding an edge between every pair of vertices at distance at most $b$ in $P_n$.

A typical scenario in which the \BANDW problem arises is that of minimizing the bandwidth of a symmetric matrix $M$ to allow for more efficient storing and manipulating procedures~\cite{Feige00survey}. The \bandw of $M$ is $b$ if all its non-zero entries are at distance at most $b$ from the diagonal. Applying permutations on the rows and columns to reduce the bandwidth of $M$ corresponds then to reordering the vertices of a graph whose adjacency matrix corresponds to $M$ by replacing all non-zero entries by $1$.

The \BANDW problem is NP-hard~\cite{Papadimitriou76}, even for trees of maximum degree at most three~\cite{GareyGJK78}, caterpillars with hair length at most three~\cite{Monien86}, and convex trees \cite{ShresthaTU11}. Even worse, approximating the \bandw within a constant factor is NP-hard, even for caterpillars of degree three~\cite{Unger98}. Further, it is known that the standard parameterization of the problem is hard for every fixed level of the W-hierarchy~\cite{BodlaenderFH94} and unlikely to be solvable in $f(b) n^{o(b)}$ time \cite{ChenHKX04}.

Faced with this immense intractability, several approaches have been proposed in the literature for the \BANDW problem.

\emph{Polynomial time approximation algorithms}.
The first (polynomial time) approximation algorithm with a polylogarithmic approximation factor was provided by Feige~\cite{Feige00approx}. Later, Dunagan and Vempala gave an $\cO(\log^3 n \sqrt{\log \log n})$-approximation algorithm. The current best approximation algorithm achieves an $\cO(\log^3 n(\log \log n)^{1/4})$-approxima\-tion factor~\cite{jrlee}.
For large $b$, the best approximation algorithm is the probabilistic algorithm of Blum {\em et al.}~\cite{BlumKRV00} which has an $\cO(\sqrt{n/b} \log n)$-approximation factor. 

\emph{Super-polynomial time approximation algorithms}.
Super-polynomial time approximation algorithms for the \BANDW problem have also been widely investigated~\cite{CyganKPW08,CyganP09icalp,DunaganV01,FeigeT05}. 
Feige and Talwar~\cite{FeigeT05}, and Cygan and Pilipczuk~\cite{CyganP09icalp} provided subexponential time approximation schemes for approximating the \bandw of graphs with small treewidth. For general graphs, a $2$-approximation algorithm with running time $\cO^*(3^n)$\footnote{$\cO^*(f(n))$ denotes $n^{\cO(1)} \cdot f(n)$.} is easily obtained by combining ideas from \cite{Feige00survey} and \cite{FeigeT05} (as noted in~\cite{CyganKPW08}). Further, Cygan {\em et al.}~\cite{CyganKPW08} provide a $3$-approximation algorithm with running time $\cO^*(2^n)$, which they generalize to a $(4r-1)$-approximation algorithm (for any positive integer $r$) with running time $\cO^*(2^{n/r})$.

\emph{Exact exponential time algorithms}.
Concerning exact exponential time algorithms, the first non-trivial algorithm was the elegant polynomial space $\cO^*(10^n)$ time algorithm of Feige~\cite{Feige00survey}. This bound has been improved in a sequence of algorithms by Cygan and Pilipczuk; their $\cO(9.363^n)$ time algorithm uses polynomial space~\cite{CyganP10}, their $\cO^*(5^n)$ time algorithm uses $\cO^*(2^n)$ space~\cite{CyganP09}, their $\cO(4.83^n)$ time algorithm uses $\cO^*(4^n)$ space~\cite{CyganP09}, and their $\cO(4.473^n)$ time algorithm uses $\cO(4.473^n)$ space ~\cite{CyganP09icalp}. The \BANDW problem can also be solved exactly in $\cO(n^b)$ time using dynamic programming~\cite{GurariS84,MonienS80}.

\emph{Graph classes}. Polynomial time algorithms for \BANDW are only known for a small number of restricted graph classes. These are caterpillars of hair length at most 2 \cite{AssmannPSZ81}, chain graphs \cite{KloksKM98}, cographs \cite{Yan97}, interval graphs \cite{KleitmanV90,MaheshRS91,Sprague94}, and bipartite permutation graphs \cite{HeggernesKM09}.
Polynomial time, constant factor approximation algorithms are known for AT-free graphs \cite{KloksKM99}, convex bipartite graphs \cite{ShresthaTU11}, and 2-directional orthogonal ray graphs \cite{ShresthaTU11}.

\emph{Hybrid algorithms}.
Another recent approach to cope with the intractability of \BANDW is through the concept of hybrid algorithms, introduced by Vassilevska {\em et al.}~\cite{VassilevskaWW06}. They gave an algorithm that after a polynomial time test, either computes the minimum \bandw of a graph in $\cO^*(4^{n+o(n)})$ time, or computes an $\cO(\gamma(n) \log^2 n \log \log n)$-approximation in polynomial time, for any unbounded constructible function $\gamma(n)$. This result was improved by Amini {\em et al.}~\cite{AminiFS09} who give an algorithm which, after a polynomial time test, either computes the minimum \bandw of a graph in $\cO^*(4^n)$ time, or provides an $\cO(\log^{3/2} n)$-approximation in polynomial time.

\medskip
The concept of designing approximation algorithms with better performance ratios at the expense of super-polynomial running times is quite natural and has been used in the study of several other problems (see for example \cite{BjorklundH06,BourgeoisEP08,BourgeoisEP08coloring,BourgeoisEP09,BourgeoisLMP10,DantsinGHK01,DyerFKKPV93,Hirsch03,JerrumV96,RamanSS07}). A similar concept is parameterized approximation, which was introduced in three independent papers~\cite{CaiH06,ChenGG06,DowneyFM06}; see the survey by Marx \cite{Marx08}.

\paragraph{Our Results.} Our main result is a $2$-approximation algorithm for the \BANDW problem that takes worst-case $\cO(1.9797^n)$ time (Theorem~\ref{thm:2approx}). This improves the $\cO^*(3^n)$ time bound achieved by Cygan {\em et al.}~\cite{CyganKPW08} for the same approximation ratio. Also, the previous best $3$-approximation algorithm of Cygan and Pilipczuk~\cite{CyganP09icalp} has an $\cO^*(2^n)$ time bound. Therefore, our $2$-approximation algorithm is also faster than the previous best $3$-approximation algorithm. 

Our algorithm is based on constructing bucket decompositions of the input  graph. A bucket decomposition partitions the vertex set of a graph into ordered sets (called {\em buckets}) of (almost) equal sizes such that all edges are either incident to vertices in the same bucket or to vertices in two consecutive buckets. The idea is to find the smallest bucket size for which there exists a bucket decomposition.  This gives a $2$-approximation for the \BANDW problem (Lemmas~\ref{lem:ub} and~\ref{lem:lb}).  The algorithm uses a simple divide-and-conquer strategy along with dynamic programming to achieve this improved time bound.



%

\section{Preliminaries}
\label{sec:prelim}
Let $G=(V,E)$ be a graph on $n$ vertices. A \emph{linear arrangement} of $G$ is a bijective function $L: V \rightarrow [n] = \{1,\hdots,n\}$, that is a numbering of its vertices from $1$ to $n$. The \emph{stretch} of an edge $u v$ is the absolute difference between the numbers assigned to its endpoints $|L(u)-L(v)|$. The \emph{\bandw\/} of a linear arrangement of $G$ is the maximum stretch over all the edges of $G$ and the \emph{\bandw\/} of $G$ is the minimum \bandw over all linear arrangements of $G$.

A \emph{bucket arrangement} of $G$ is a placement of its vertices into buckets such that for each edge, its endpoints are either in the same bucket or in two consecutive buckets~\cite{FeigeT05}. The buckets are linearly ordered and numbered from left to right. 
A \emph{capacity vector} $\CCC$ is a vector of positive integers. The \emph{length} of a capacity vector $\CCC = (\CCC[1], \hdots, \CCC[k])$ is $k$ and its \emph{size} is $\sum_{i=1}^k \CCC[i]$.
Given a capacity vector $\CCC$ of size $n$, a \emph{$\CCC$-bucket arrangement} of $G$ is a bucket arrangement in which exactly $\mathcal{C}[i]$ vertices are placed in bucket~$i$, for each $i$. For integers $n$ and $\ell$ with $\ell < n/2$, an \emph{$(n,\ell)$-capacity vector} is a capacity vector
$$( a, \underbrace{\ell, \ell, \hdots, \ell,}_{ \left\lceil \frac{n}{\ell} \right\rceil -2 \textrm{ times}} b )$$
of size $n$ such that $a,b \le \ell$. 
We say that an $(n,\ell)$-capacity vector is {\em left-packed} if $a = \ell$ and {\em balanced} if $|a-b|\le 1$. 

Let $X \subseteq V$ be a subset of the vertices of $G$. We denote by $G[X]$ the subgraph of $G$ induced on $X$, and by $G \setminus X$ the subgraph of $G$ induced by $V\setminus X$. The \emph{open neighborhood} of a vertex $v$ is denoted by $N_G(v)$ and the \emph{open neighborhood} of $X$ is $N_G(X):=(\bigcup_{v \in X} N_G(v)) \setminus X$. 

\section{Exponential Time Algorithms for Approximating Bandwidth}
We first establish two simple lemmas that show that constructing a bucket arrangement can approximate the \bandw of a graph. 
\begin{lemma}\label{lem:ub}
Let $G$ be a graph on $n$ vertices, and let $\CCC$ be an $(n,\ell)$-capacity vector. If there exists a $\mathcal{C}$-bucket arrangement for $G$ then the \bandw of $G$ is at most $2\ell - 1$.
\end{lemma}
\begin{proof}
Given a $\mathcal{C}$-bucket arrangement for $G$, create a linear arrangement $L$ respecting the bucket arrangement (if $u$ appears in a smaller numbered bucket than $v$, then $L(u) < L(v)$), where vertices in the same bucket are numbered in an arbitrary order. As the capacity of each bucket is at most $\ell$ and each edge spans at most two consecutive buckets, the maximum edge stretch in $L$ is at most $2 \ell -1$.
\end{proof}

\begin{lemma}\label{lem:lb}
Let $G$ be a graph on $n$ vertices, and let $\CCC$ be an $(n,\ell)$-capacity vector. If there exists no $\mathcal{C}$-bucket arrangement for $G$ then the \bandw of $G$ is at least $\ell + 1$.
\end{lemma}
\begin{proof}
Suppose there exists a linear arrangement $L$ of $G$ of \bandw at most $\ell$. Construct a bucket arrangement placing the first $\mathcal{C}[1]$ vertices of $L$ into the first bucket, the next $\mathcal{C}[2]$ vertices of $L$ into the second bucket, and so on. In the resulting bucket arrangement, no edge spans more than two consecutive buckets. Therefore, a $\mathcal{C}$-bucket arrangement exists for $G$, a contradiction.
\end{proof}

\noindent
We will use the previous fastest $2$-approximation algorithm of Cygan {\em et al.}~\cite{CyganKPW08} as a subroutine. For completeness, we describe this simple algorithm here. 

\begin{proposition}[\cite{CyganKPW08}]\label{prop:slow2approx}
There is a polynomial space $2$-approximation algorithm for the \BANDW problem that takes
$\cO^*(3^n)$ time on connected graphs with $n$ vertices.
\end{proposition}
\begin{proof}
Let $G$ be a connected graph on $n$ vertices.
For $\ell$ increasing from $1$ to $\lceil n/2 \rceil$, the algorithm does the following.
Let $\mathcal{C}$ be an $(n,\ell)$-capacity vector. The algorithm goes over all the $k=\left\lceil \frac{n}{\ell} \right\rceil$ choices for assigning the first vertex to some bucket in a $\mathcal{C}$-bucket arrangement. The algorithm then chooses an unassigned vertex $u$ which has at least one neighbor that has already been assigned to some bucket. Assume that a neighbor of $u$ is assigned to the bucket~$i$. Now there are at most three choices of buckets ($i-1$, $i$, and $i+1$) for assigning vertex $u$. Some of these choices may be invalid either because of the capacity constraints of the bucket or because of the previous assignments of (other) neighbors of $u$. If the choice is valid, the algorithm recurses by assigning $u$ to that bucket.
Let $\ell'$ be the smallest integer for which the algorithm succeeds, in some branch, to place all vertices of $G$ into buckets in this way. Then, by Lemma~\ref{lem:ub}, $G$ has \bandw at most $2\ell'-1$ and by Lemma~\ref{lem:lb}, $G$ has \bandw at least $\ell'$. Thus, the algorithm outputs $2\ell'-1$, which is a $2$-approximation for the \bandw of $G$.
As the algorithm branches into at most $3$ cases for each of the $n$ vertices (except the first one), and all other computations only contribute polynomially to the running time of the algorithm, this algorithm runs in worst-case $\cO^*(3^n)$ time using only polynomial space.
\end{proof}

\noindent
We now show another simple algorithm based on a divide-and-conquer strategy that given an $(n,\ell)$-capacity vector $\CCC$, decides whether a $\mathcal{C}$-bucket arrangement exists for a connected graph $G$.

\begin{proposition}\label{prop:generalalgo}
There is an algorithm that has as input a connected graph $G$ on $n$ vertices and an $(n,\ell)$-capacity vector $\CCC$ with $\ell < n/2$ and decides whether $G$ has a $\CCC$-bucket arrangement in $\cO^*\left( \binom{n}{\ell} \cdot \binom{n/2}{\ell} \cdot 2^{4\ell} \cdot 3^{n/4} \right)$ time.
\end{proposition}
\begin{proof}
Let $k=\left\lceil \frac{n}{\ell} \right\rceil$ be the number of buckets in a $\mathcal{C}$-bucket arrangement. Number the buckets from $1$ to $k$ from left to right according to the bucket arrangement.
We solve a slightly more general problem where some subset $Q_1$ of vertices is restricted to bucket~$1$ and some subset $Q_k$ of vertices is restricted to bucket~$k$. If $|Q_1|>\ell$ or $|Q_k|>\ell$, then answer \textsc{No}.
Select a bucket index $i$ such that the sum of the capacities of the buckets numbered strictly smaller than $i$ and the one for the buckets numbered strictly larger than $i$ are both at most $n/2$. 
 
The algorithm goes over all possible $\binom{n-|Q_1\cup Q_k|}{\ell}$ choices for filling bucket~$i$ with $\ell$ vertices from $V\setminus (Q_1\cup Q_k)$. Let $X$ be a set of $\ell$ vertices assigned to the bucket~$i$. Given a connected component of $G\setminus X$, note that all the vertices of this connected component must be placed either only in buckets $1$ to $i-1$ or buckets $i+1$ to $k$. 
Note that each connected component of $G\setminus X$ contains at least one vertex that is adjacent to a vertex in $X$ (as $G$ is connected). Therefore, for each connected component of $G\setminus X$, at least one vertex is placed into the bucket~$i-1$ or $i+1$. As the capacity of each bucket is at most $\ell$, $G\setminus X$ has at most $2 \ell$ connected components, otherwise there is no $\mathcal{C}$-bucket arrangement where $X$ is assigned to the bucket~$i$.
Thus, there are at most $2^{2 \ell}$ choices for assigning connected components of $G\setminus X$ to the buckets $1$ to $i-1$ and $i+1$ to $k$. Some of these assignments might be invalid as they might violate the capacity constraints of the buckets or assign vertices from $Q_1$ to the buckets $i+1$ to $k$ or vertices from $Q_k$ to buckets $1$ to $i-1$. We discard these invalid assignments. 
 
For each choice of $X$ and each valid assignment of the connected components of $G\setminus X$ to the left or right of bucket~$i$, we have now obtained two independent subproblems: one subproblem for the buckets $\{1,\hdots,i-1\}$ and one subproblem for the buckets $\{i+1,\hdots,k\}$. The instances of these subproblems have at most $n/2$ vertices. Consider the subproblem for the buckets $\{1,\hdots,i-1\}$ (the other one is symmetric) and let $Y$ be the set of vertices associated to these buckets. Let $Z \subseteq Y$ be the set of vertices in $Y$ that have at least one neighbor in $X$. Now, restrict $Z$ to bucket~$i-1$ and add edges to the subgraph $G[Y]$ such that $Z$ becomes a clique. This does not change the problem, as all the vertices in $Z$ must be assigned to the bucket~$i-1$, and $G[Y]$ becomes connected. This subproblem can be solved recursively.

The algorithm performs the above recursion until it reaches subproblems of size at most $n/4$, which corresponds to two levels in the corresponding search tree. On instances of size at most $n/4$, the algorithm invokes the algorithm of Proposition~\ref{prop:slow2approx}, which takes worst-case $\cO^*(3^{n/4})$ time. 
 
Let $T(n)$ be the running time needed for the above procedure to check whether a graph with $n$ vertices has a bucket arrangement for an $(n,\ell)$-capacity vector.  Then,
\begin{align*}
T(n) &\leq \binom{n}{\ell} \cdot 2^{2 \ell} \cdot \binom{n/2}{\ell} \cdot 2^{2 \ell} \cdot 3^{n/4} \cdot n^{\cO(1)} = \cO^*\left( \binom{n}{\ell} \cdot \binom{n/2}{\ell} \cdot 2^{4\ell} \cdot 3^{n/4} \right).
\end{align*}
This completes the proof of the proposition.
\end{proof}

\noindent
Combining Proposition~\ref{prop:generalalgo} with Lemmas~\ref{lem:ub} and \ref{lem:lb}, we have the following corollary for $2$-approximating the \bandw of a graph.

\begin{corollary}\label{cor:smallbandwidth}
There is an algorithm that, for a connected graph $G$ on $n$ vertices and an integer $\ell \le n$, decides whether the \bandw of $G$ is at least $\ell+1$ or at most $2\ell-1$ in $\cO^*\left( \binom{n}{\ell} \cdot \binom{n/2}{\ell} \cdot 2^{4\ell} \cdot 3^{n/4} \right)$ time.
\end{corollary}
\begin{proof}
If $\ell \ge n/2$, the \bandw of $G$ is at most $2 \ell -1$. Otherwise, use Proposition~\ref{prop:generalalgo} with $G$ and some $(n,\ell)$-capacity vector $\mathcal{C}$ to decide if there exists a $\mathcal{C}$-bucket arrangement for $G$. If so, then the bandwidth of $G$ is at most $2 \ell - 1$ by Lemma~\ref{lem:ub}. If not, then the bandwidth of $G$ is at least $\ell + 1$ by Lemma~\ref{lem:lb}.
\end{proof}

\noindent
The running time of the algorithm of Corollary~\ref{cor:smallbandwidth} is interesting for small values of $\ell$. Namely, if $\ell \leq n/26$, the running time is $\cO(1.9737^n)$. In the remainder of this section, we improve Proposition~\ref{prop:generalalgo}. 
We now concentrate on the cases where $k = \lceil n/\ell \rceil  \leq 26$.

Let $\CCC$ be an $(n,\ell)$-capacity vector.  A {\em partial $\CCC$-bucket arrangement} of an induced subgraph $G'$ of $G$ is a placement of vertices of $G'$ into buckets such that: (a) each vertex in $G'$ is assigned to a bucket or to the union of two consecutive buckets (i.e., the vertex is restricted to belong to one of these buckets, but it is not fixed to which one of these two buckets the vertex belongs), and each vertex that is assigned to the union of two consecutive buckets can be assigned to one of these two buckets such that (b) the endpoints of each edge in $G'$ are either in the same bucket or in two consecutive buckets, and (c) at most $\CCC[i]$ vertices are placed in each bucket~$i$. A vertex that is assigned to a bucket is not assigned to the union of two buckets, and vice-versa. Let $\BBB$ be a partial $\CCC$-bucket arrangement of an induced subgraph $G'$. 
We say that a bucket~$i$ is \emph{full} in $\BBB$ if the number of vertices that have been assigned to it equals its capacity $\CCC[i]$. We say that two consecutive buckets $i$ and $i+1$ are \emph{jointly full} in $\BBB$ if a vertex subset $Y$ of cardinality equal to the sum of the capacities of $i$ and $i+1$ has been assigned to the union of these two buckets. We always maintain the condition that if a bucket is full then no vertex has been assigned to the union of this bucket and some other bucket, and if two consecutive buckets are jointly full then no vertex has been assigned to any one of these two buckets individually. We say that a bucket is {\em empty} in $\BBB$ if no vertex has been assigned to it nor to the union of this bucket and a neighboring bucket.
\newenvironment{proofidea}{{\noindent\it Proof Outline.}}{\hfill{\qed}\vspace{1ex}}
\begin{proposition}\label{prop:dp}
Let $G$ be a graph on $n$ vertices and $\mathcal{C}$ be a capacity vector of size $n$ and length $k$, where $k$ is an integer constant. Let $\mathcal{B}$ be a partial $\CCC$-bucket arrangement of some induced subgraph $G'$ of $G$ such that in $\BBB$ some buckets are full, some pairs of consecutive buckets are jointly full, and all other buckets are empty.  If in $\BBB$ no $3$ consecutive buckets are empty, then it can be decided if $\mathcal{B}$ can be extended to a $\mathcal{C}$-bucket arrangement in polynomial time.
\end{proposition}
\begin{proof}
Let $G=(V,E)$ and $G'=(V',E')$.  Let $r'$ be the number of connected components of $G \setminus V'$ (the graph induced by $V \setminus V'$), and let $V_{\ell}'$ represent the set of vertices in the $\ell$th connected component of $G \setminus V'$.

If the bucket~$i$ is full in $\BBB$, let $X_i$ denote the set of vertices assigned to it. If the buckets $i$ and $i+1$ are jointly full in $\BBB$, let $X_{i,i+1}$ denote the set of vertices assigned to the union of buckets $i$ and $i+1$.

We use dynamic programming to start from a partial bucket arrangement satisfying the above conditions to construct a $\CCC$-bucket arrangement.
During its execution, the algorithm assigns vertices to the buckets which are empty in $\BBB$.
A vertex is always assigned to the union of two consecutive empty buckets or to a single empty bucket.
It can only be assigned to a single empty bucket if that bucket has no neighboring empty bucket.
The idea is to iteratively assign the $V_{\ell}'$ to empty buckets and to maintain only a count of the number of vertices constrained to buckets in various ways.

Note that if $V_{\ell_1}'$ and $V_{\ell_2}'$ have a common neighbor in $X_{i,i+1}$, then $V_{\ell_1}'$ and $V_{\ell_2}'$ need to be assigned to the same bucket(s). On the other hand,
in order to determine how many vertices from $X_{i,i+1}$ are constrained to bucket~$i$, we cannot treat $V_{\ell_1}'$ and $V_{\ell_2}'$ separately.
Obtain $\mathcal{V}=\{V_1, \dots, V_r\}$ from $\mathcal{V}' = \{V_1', \dots, V_{r'}'\}$ by repeatedly merging $V_{\ell_1}',V_{\ell_2}'\in \mathcal{V}'$ if they have a common neighbor
in two consecutive buckets that are jointly full.

The dynamic programming algorithm constructs a table $T$, which has the following indices.
\begin{itemize}
\item An index $p$, representing the subproblem constrained to $V_1, \dots, V_p$.
\item For every empty bucket~$i$ in $\BBB$ such that neither the bucket~$i-1$ nor the bucket~$i+1$ is empty, it has an index $s_i$, representing the number of vertices assigned to the bucket~$i$.
\item For every two consecutive empty buckets $i$ and $i+1$ in $\BBB$, it has indices $t_{i,i+1}$, $x_i$, and $x_{i+1}$. The index $t_{i,i+1}$ represents the total number of vertices assigned to the buckets $i$ and $i+1$. The index $x_i$ represents the number of vertices assigned to the buckets $i$ and $i+1$ that have at least one neighbor in the bucket~$i-1$. The index $x_{i+1}$ represents the number of vertices assigned to the buckets $i$ and $i+1$ that have at least one neighbor in the bucket~$i+2$.
\item For every two consecutive buckets $i$, $i+1$ which are jointly full in $\BBB$, it has indices $f_i$ and $f_{i+1}$ representing the number of vertices assigned to these buckets that have at least one neighbor in the bucket~$i-1$ ($f_i$) or in the bucket~$i+2$ ($f_{i+1}$).
\end{itemize}
The table $T$ is initialized to $\mathsf{false}$ everywhere, except for the entry corresponding to all-zero indices, which is initialized to $\mathsf{true}$. The rest of the table is built by increasing values of $p$ as described below. Here, we only write those indices that differ in the looked-up table entries and the computed table entry (i.e., indices in the table that play no role in a given recursion are omitted).  We also ignore the explicit checking of the invalid indices in the following description. The algorithm looks at the vertices which are neighbors (in $G$) of the vertices in $V_p$ and have already been assigned.
 
If $N_G(V_p)$ contains a vertex from each of the full buckets $i-1$ and $i+1$, no vertex from any other bucket, and bucket~$i$ is empty in $\BBB$, then
\begin{align*}
T[p,s_i,\hdots]=T[p-1,s_i-|V_p|,\hdots].
\end{align*}
If $N_G(V_p)$ contains a vertex from the full buckets $i-1$ and $i+2$, no vertex from any other bucket, and the buckets $i$ and $i+1$ are both empty in $\BBB$, then 
\begin{align*}
&T[p,t_{i,i+1},x_i,x_{i+1},\hdots] = \\
&\quad\;
\begin{cases}
\mathsf{false} & \hspace*{-.1in}\text{if } N_G(X_{i-1}) \cap N_G(X_{i+2}) \neq \emptyset,\\
T[p-1,t_{i,i+1}-|V_p|,x_i-|V_p\cap N_G(X_{i-1})|, \\
 \; \; \, x_{i+1}-|V_p\cap N_G(X_{i+2})|,\hdots] & \hspace*{-.1in}\text{otherwise}.
\end{cases}
\end{align*}
If $N_G(V_p)$ contains a vertex from the jointly full buckets $i-2$ and $i-1$ and a vertex from the jointly full buckets $i+1$ and $i+2$, but no vertex from any other bucket, and bucket~$i$ is empty in $\BBB$, then 
\begin{align*}
T[p,s_i,f_{i-1},f_{i+1},\dots] = T[&p-1,s_i-|V_p|,f_{i-1}-|N_G(V_p)\cap X_{i-2,i-1}|,\\&f_{i+1}-|N_G(V_p)\cap X_{i+1,i+2}|,\dots].
\end{align*}
In the above recursion, observe that, by the definition of $\mathcal{V}$, no $V_q \in \mathcal{V} \setminus V_p$ has a common neighbor with $V_p$ in $G$.
Therefore, the vertices counted towards $f_{i-1}$ and $f_{i+1}$ will not be recounted towards $f_{i-1}$ and $f_{i+1}$ for any $p'>p$.

The recursion for the other possibilities where $V_p$ has neighbors in (at least) two distinct buckets are similar and can easily be deduced. We now consider the cases where $V_p$ has only neighbors in one bucket. Again, we only describe some key-cases, from which all other cases can easily be deduced.

If the vertices in $V_p$ have only neighbors in the full bucket~$i-1$, and the buckets $i-2$ and $i$ are both empty in $\BBB$, but the buckets $i-3$ and $i+1$ are either full or non-existing, then
 \begin{align*}
 T[p,s_{i-2},s_i,\hdots]=T[p-1,s_{i-2}-|V_p|,s_i,\hdots]~\vee~T[p-1,s_{i-2},s_i-|V_p|,\hdots].
 \end{align*}
If the vertices in $V_p$ have only neighbors in the full bucket~$i-1$, and the buckets $i-3$, $i-2$, $i$, and $i+1$ are all empty in $\BBB$, then
 \begin{align*}
  &T[p,t_{i-3,i-2},x_{i-2},t_{i,i+1},x_i,\hdots] =\\
  &\quad\quad\quad T[p-1,t_{i-3,i-2}-|V_p|,x_{i-2}-|V_p\cap N_G(X_{i-1})|,t_{i,i+1},x_i,\hdots]\\
  &\quad\quad\quad \vee~T[p-1,t_{i-3,i-2},x_{i-2},t_{i,i+1}-|V_p|,x_i-|V_p\cap N_G(X_{i-1})|,\hdots].
 \end{align*}
If the vertices in $V_p$ have only neighbors in the jointly full buckets $i$ and $i+1$, and the buckets $i-1$ and $i+2$ are both empty in $\BBB$, but the buckets $i-2$ and $i+3$ are either full in $\BBB$ or non-existing, then
 \begin{align*}
  &T[p,s_{i-1},s_{i+2},f_i,f_{i+1},\hdots] =\\
  &\quad\quad\quad T[p-1,s_{i-1}-|V_p|,s_{i+2},f_i-|N_G(V_p)\cap X_{i,i+1}|,f_{i+1},\hdots]\\
  &\quad\quad\quad \vee~T[p-1,s_{i-1},s_{i+2}-|V_p|,f_i,f_{i+1}-|N_G(V_p)\cap X_{i,i+1}|,\hdots].
 \end{align*}
Again, recall that no $V_q \in \mathcal{V} \setminus V_p$ has a common neighbor with $V_p$ in $G$.
Therefore, the vertices counted towards $f_{i}$ and $f_{i+1}$ will not be counted towards $f_{i}$ and $f_{i+1}$ again for any $p'>p$.
The final answer ($\mathsf{true}$ or $\mathsf{false}$) produced by the algorithm is a disjunction over all table entries whose indices are as follows:
  $p=r$,
  $s_i=\CCC[i]$ for every index $s_i$,
  $t_{i,i+1}=\CCC[i]+\CCC[i+1]$ for every index $t_{i,i+1}$,
  $x_i \le \CCC[i]$ for every index $x_i$, and
  $f_i \le \CCC[i]$ for every index $f_i$.

Since the number of relevant table entries is $O(n^{3k/2+1})$ and each entry can be computed in linear time, the running time of this algorithm is polynomial in $n$ for a constant $k$.
\end{proof}

\begin{remark}
The dynamic programming algorithm in Proposition~\ref{prop:dp} can easily be modified to construct a $\CCC$-bucket arrangement (from any partial bucket arrangement $\BBB$ satisfying the stated conditions), if one exists.
\end{remark}

If the number of buckets is a constant, the following proposition will be crucial in speeding up the procedure for assigning connected components to the right or the left of a bucket filled with a vertex set $X$. Denote by $\smallcomp{G}$ the set of all connected components of $G$ with at most $\sqrt{n}$ vertices and by $\largecomp{G}$ the set of all connected components of $G$ with more than $\sqrt{n}$ vertices. Let $\smallcompvert{G}$ and $\largecompvert{G}$ denote the set of all vertices which are in the connected components belonging to $\smallcomp{G}$ and $\largecomp{G}$, respectively. We now make use of the fact that if there are many small components in $G\setminus X$, several of the assignments of the vertices in $\smallcompvert{G\setminus X}$ to the buckets are equivalent.

A partial $\CCC$-bucket arrangement is \emph{pure} if it does not assign a vertex to the union of two consecutive buckets.
Let $\CCC$ be a capacity vector of size $n$ (i.e., $\sum_i \CCC[i]=n$) and let $\BBB$ be a pure partial $\CCC$-bucket arrangement of an induced subgraph $G'$ of $G$. We say that $\BBB$ \emph{produces} the capacity vector $\CCC'$ if $\CCC'$ is obtained from $\CCC$ by decreasing the capacity $\CCC[i]$ of each bucket~$i$ by the number of vertices assigned to the bucket~$i$ in $\BBB$.


\begin{proposition}\label{prop:leftright}
Let $G=(V,E)$ be a graph on $n$ vertices. Let $\mathcal{C}$ be a capacity vector of size $n$ and length $k$, where $k$ is an integer constant. Let $j$ be a bucket and $X\subseteq V$ be a subset of $\mathcal{C}[j]$ vertices. Consider all capacity vectors that are produced by the pure partial $\CCC$-bucket arrangements of $G[\smallcompvert{G \setminus X} \cup X]$ in which the vertices in $X$ are assigned to the bucket~$j$. Then, there exists an algorithm which runs in $\cO^*(3^{\sqrt{n}})$ time and takes polynomial space, and enumerates all $($distinct$)$ capacity vectors produced by these pure partial $\CCC$-bucket arrangements.
\end{proposition}
\begin{proof}
Let  $V_l$ be the vertex set of the $l$th connected component in $\smallcomp{G \setminus X}$. Let $\mathcal{L}_{p}$ denote the list of all capacity vectors
produced by the pure partial $\CCC$-bucket arrangements of
$G[\bigcup_{1 \le l \le p} V_l \cup X]$ in which the vertices in $X$ are assigned to the bucket~$j$. Note that
the number of distinct vectors in $\mathcal{L}_p$ is $O(n^k)$.
Then, $\mathcal{L}_1$ can be obtained by executing the algorithm of Proposition~\ref{prop:slow2approx} on the graph $G[V_1]$ with a capacity vector $\CCC'$ which is the same as $\CCC$ except that $\CCC'[i]=0$. In general, $\mathcal{L}_p$ can be obtained from $\mathcal{L}_{p-1}$ by executing the algorithm of Proposition~\ref{prop:slow2approx} on the graph $G[V_p]$ for every capacity vector in $\mathcal{L}_{p-1}$. As the size of each connected component in $\smallcomp{G\setminus X}$ is at most $\sqrt{n}$, the resulting running time is $\cO^*(3^{\sqrt{n}})$.
\end{proof}


\subsection{Exponential Time $2$-Approximation Algorithm for \BANDW}
Let $G=(V,E)$ be the input graph. Our algorithm tests all bucket sizes $\ell$ from $1$ to $\lceil n/2 \rceil$ until it finds an $(n,\ell)$-capacity vector $\mathcal{C}$ such that $G$ has a $\mathcal{C}$-bucket arrangement. For a given $\ell$, let $k=\left\lceil \frac{n}{\ell} \right\rceil$ denote the number of buckets. Our algorithm uses various strategies depending on the value of $k$. The case of $k=1$ is trivial. If $\ell = \lceil n/2 \rceil$, we have at most two buckets and any partition of the vertex set of $G$ into sets of sizes $\ell$ and $n-\ell$ is a valid $\mathcal{C}$-bucket arrangement. If $k \ge 27$, Corollary~\ref{cor:smallbandwidth} gives a running time of $\cO(1.9737^n)$. For all other values of $k$, we will obtain
running times in $\cO(1.9797^n)$.

Let $I_k$ be the set of all integers lying between $n/k)$ and $n/(k-1)$. We have that $\ell\in I_k$. The basic idea (as illustrated in Proposition~\ref{prop:generalalgo}) is quite simple. The algorithm tries all possible ways of assigning vertices to the middle bucket. Once the vertex set $X$ assigned to the middle bucket is fixed and the algorithm has decided for each connected component of $G\setminus X$ if the connected component is to be assigned to the buckets to the left or to the right of the middle bucket, the problem breaks into two independent subproblems on buckets which are to the left and to right of the middle bucket. To get the claimed running time, we build upon this idea to design individualized techniques for different $k$s (between $3$ and $26$). For each case, if $G$ has at least one $\CCC$-bucket arrangement for an $(n,\ell)$-capacity vector $\CCC$, then one such arrangement is constructed. We know that if $G$ has no $\CCC$-bucket arrangement for an $(n,\ell)$-capacity vector $\CCC$ then the \bandw of $G$ is at least $\ell+1$ (Lemma~\ref{lem:lb}), and if it has one then its \bandw is at most $2\ell-1$ (Lemma~\ref{lem:ub}). 
If $k=8,10,$ or $12$, the algorithm uses a left-packed $(n,\ell)$-capacity vector $\CCC$, and otherwise, the algorithm uses a balanced $(n,\ell)$-capacity vector $\CCC$.

\paragraph{$\mathbf{k=3}$.}  The algorithm goes over all subsets $X \subseteq V$ of cardinality $|X|= \CCC[3] \le \lceil (n-\ell)/2 \rceil$. $X$ is assigned to the bucket~$3$. If the remaining vertices can be assigned to the buckets $1$ and $2$ in a way such that all vertices which are neighbors of the vertices in $X$ (in $G$) are assigned to the bucket~$2$, then $G$ has a $\CCC$-bucket arrangement where $\CCC$ has length $3$. 
The worst-case running time for this case is $\max_{\, \ell \in I_3} \cO^*( \binom{n}{|X|})$.

\paragraph{$\mathbf{k=4}$ or $\mathbf{k=5}$.} The algorithm goes over all subsets $X \subseteq V$ with $|X|=\ell$. $X$ is assigned to the bucket~$3$.  Then, we can conclude using the dynamic programming algorithm from Proposition~\ref{prop:dp} (see also the remark following it). 
The worst-case running time for these cases are $\max_{\, \ell \in I_k} \cO^*( \binom n \ell)$.


\paragraph{$\mathbf{k=6}$.} If $k=6$, the algorithm goes through all subsets $X \subseteq V$ with $|X|=2\ell$. $X$ is assigned to the union of buckets $3$ and $4$ (i.e., some non-specified $\ell$ vertices from $X$ are assigned to the bucket~$3$, and the remaining vertices of $X$ are assigned to the bucket~$4$). Then, we can again conclude by the algorithm from Proposition~\ref{prop:dp}. The worst-case running time for this case is $\max_{\, \ell \in I_6} \cO^*\left(\binom n {2\ell}\right)$.

\begin{table}[htp]
\begin{small}
\begin{eqnarray*}
\hline\\
k & ~~\mbox{Running time}~~ & \mbox{Expression}\\
\hline \\
k \leq 2 & \mbox{poly}(n) & \\
k = 3 & \cO(1.8899^n) & \max_{\ell \in I_3} \left \{ \binom n {\frac{n-\ell}{2}} \right\} = \binom{n}{\frac n 3}\\
k = 4 & \cO(1.8899^n) & \max_{\ell \in I_4} \left\{\binom{n}{\ell} \right\} = \binom{n}{\frac n 3}\\
k = 5 & \cO(1.7548^n) & \max_{\ell \in I_5}\left \{ \binom n \ell \right \} = \binom{n}{\frac n 4}\\
k = 6 & \cO(1.9602^n) & \max_{\ell \in I_6} \left \{ \binom n {2\ell} \right \} = \binom{n}{\frac{2n}{5}}\\
k = 7 & \cO(1.9797^n) & \max_{\ell \in I_7} \left\{ {n \choose \ell} \cdot 2^{\cO(\sqrt{n})} \cdot {\frac{(n-\ell)}{2} \choose \frac{(n-5\ell)}{2} } \right\} = \binom{n}{\frac n 7} \cdot \binom{\frac{3n}{7}}{\frac n 7} \cdot 2^{\cO(\sqrt{n})}\\
k = 8 & \cO(1.9797^n) & \max_{\ell \in I_8} \left \{ \binom{n}{\ell} \cdot 2^{\cO(\sqrt{n})} \cdot \max \left  \{\binom{3\ell}{\ell}, \binom{n-4\ell}{\ell} \right \} \right\} = \binom{n}{\frac n 7} \cdot \binom{\frac{3n}{7}}{\frac n 7} \cdot 2^{\cO(\sqrt{n})} \\
k = 9 & \cO(1.8937^n) & \max_{\ell \in I_9} \left\{ {n \choose \ell} \cdot 2^{\cO(\sqrt{n})} \cdot {\frac{(n-\ell)}{2} \choose \ell } \right\} = \binom{n}{\frac n 8} \cdot \binom{\frac{7n}{18}}{\frac n 8} \cdot 2^{\cO(\sqrt{n})}\\
k = 10 & \cO(1.8199^n) & \max_{\ell \in I_{10}} \left\{ {n \choose \ell } \cdot 2^{\cO(\sqrt{n})} \cdot \max \left\{ \binom{4\ell}{\ell}, \binom{n-5\ell}{\ell} \right\} \right\} = \binom{n}{\frac n 9} \cdot \binom{\frac{4n}{9}}{\frac n 9} \cdot 2^{\cO(\sqrt{n})} \\
k = 11 & \cO(1.7568^n) & \max_{\ell \in I_{11}} \left\{ {n \choose \ell} \cdot 2^{\cO(\sqrt{n})} \cdot {\frac{(n-\ell)}{2} \choose \ell} \right\} = \binom{n}{\frac{n}{10}} \cdot \binom{\frac{9n}{20}}{\frac{n}{10}} \cdot 2^{\cO(\sqrt{n})}\\
k = 12 & \cO(1.8415^n) & \max_{\, \ell \in I_{12}} \left\{ {n \choose \ell} \cdot 2^{\cO(\sqrt{n})} \cdot \max \left\{ {5\ell \choose \ell}, {n-6\ell \choose 2\ell} \right\} \right\} = \binom{n}{\frac{n}{11}} \cdot \binom{\frac{5n}{11}}{\frac{2n}{11}} \cdot 2^{\cO(\sqrt{n})}\\
13 \le k \le 23 & \cO(1.9567^n) & \max_{\ell \in I_{k}} \left\{ {n \choose \ell} \cdot {n/2 \choose \ell} \cdot {n/4 \choose \ell} \cdot 2^{\cO(\sqrt{n})} \right\} = \binom{n}{\frac{n}{12}} \cdot \binom{\frac{n}{2}}{\frac{n}{12}} \cdot \binom{\frac{n}{4}}{\frac{n}{12}} \cdot 2^{\cO(\sqrt{n})}\\
24 \le k \le 26 & \cO(1.6869^n) & \max_{\ell \in I_{k}} \left\{ {n \choose \ell} \cdot {n/2 \choose \ell} \cdot {n/4 \choose \ell} \cdot {n/8 \choose \ell} \cdot 2^{\cO(\sqrt{n})} \right\} = 2^{\cO(\sqrt{n})} \cdot \prod_{i=0}^3 \binom{\frac{n}{2^i}}{\frac{n}{23}}\\
k \ge 27 & \cO(1.9737^n) & \max_{\ell \in I_{k}} \left\{ {n \choose \ell} \cdot {n/2 \choose \ell} \cdot 2^{4\ell} \cdot 3^{\frac{n}{4}} \right\} =  \binom{n}{\frac{n}{26}} \cdot \binom{\frac{n}{2}}{\frac{n}{26}} \cdot 2^{\frac{2n}{13}} \cdot 3^{\frac{n}{4}}\\
\hline
\end{eqnarray*}
\caption{\label{tab:runtimes} Running time of the $2$-approximation algorithm for \BANDW according to the number of buckets $k=\lceil n/\ell \rceil$. $I_k$ is the set of all integers lying between $n/(k-1)$ and $n/k$. The final running time is dominated by the cases of $k=7$ and $k=8$ (when $\ell$ is close to $n/7$).}
\end{small}
\end{table}

\paragraph{$\mathbf{k=7}$.} The algorithm goes through all subsets $X \subseteq V$ with $|X|=\ell$. $X$ is assigned to the bucket~$4$. For each such $X$, the algorithm uses Proposition~\ref{prop:leftright} to enumerate all possible capacity vectors produced by the pure partial $\CCC$-bucket arrangements of $G[\smallcompvert{G \setminus X} \cup X]$ (with $X$ assigned to the bucket~$4$). This step can be done in $\cO^*(3^{\sqrt{n}})$ time. There are only polynomially many such (distinct) capacity vectors. For each such capacity vector $\CCC'$, the algorithm goes through all choices of assigning each connected component in $\largecomp{G \setminus X}$ to the buckets $1$ to $3$ or to the buckets $5$ to $7$. Thus, we obtain two independent subproblems on the buckets $1$ to $3$ and on the buckets $5$ to $7$. As the number of components in $\largecomp{G \setminus X}$ is at most $\sqrt{n}$ (as each connected component has at least $\sqrt{n}$ vertices), going through all possible ways of assigning each connected component in $\largecomp{G\setminus X}$ to the buckets numbered smaller or larger than $4$ takes $\cO^*(2^{\sqrt{n}})$ time.  Some of these assignments may turn out to be invalid.
For each valid assignment, let $V_1$ denote the vertex set assigned to the buckets $1$ to $3$. Then, the vertices of $V_1$ are assigned to the buckets $1$ to $3$ as described in the case with $3$ buckets with the capacity vector $(\CCC'[1], \CCC'[2], \CCC'[3])$ and with the additional restriction that all vertices in $V_1$ which are neighbors of the vertices in $X$ need to be assigned to the bucket~$3$. The number of vertices in $V_1$ is at most $\lceil (n-\ell)/2 \rceil$ (as $\CCC$ is balanced). Now the size of bucket~$1$ is $\CCC'[1] \le \lceil (n-5\ell)/2 \rceil$. Let $n_1= \lceil (n-\ell)/2 \rceil$ and $\ell_1=\lceil (n-5\ell)/2 \rceil$. Since $n_1\ge 2 \ell_1$, there are at most $\binom{n_1}{\ell_1}$ choices to assign a subset of $V_1$ to bucket~$1$. If $V_1$ has at least one valid bucket arrangement into $3$ buckets (with vertices in $V_1$ neighboring the vertices in $X$ assigned to the bucket~$3$), then the above step will construct one in worst-case $\cO^*(\binom{n_1}{\ell_1})$ time. The algorithm uses a similar approach for $V_2=V\setminus (V_1 \cup X)$ with the buckets $5$ to $7$. Since, the algorithm tries out every subset $X$ for bucket~$4$, the worst-case running time for this case is 
\begin{equation*}
\max_{\, \ell \in I_7}\cO^*\left( \binom{n}{\ell} \cdot \left( 3^{\sqrt{n}} + 2^{\sqrt{n}} \cdot \binom{n_1}{\ell_1} \right) \right) = \max_{\, \ell \in I_7}\cO^*\left( \binom{n}{\ell} \cdot 2^{\cO(\sqrt{n})} \cdot \binom{n_1}{\ell_1} \right) \enspace.
\end{equation*}

\paragraph{$\mathbf{k=8}$.} The algorithm uses a left-packed $(n,\ell)$-capacity vector $\CCC$ for this case. The algorithm goes through all subsets $X \subseteq V$ with $|X|=\ell$. $X$ is assigned to the bucket~$4$. The remaining analysis is similar to the case with $7$ buckets. 
The algorithm considers each capacity vector $\CCC'$ produced by the pure partial $\CCC$-bucket arrangements of $G[\smallcompvert{G \setminus X} \cup X]$, where $X$ is assigned to the bucket~$4$,
and each valid choice of assigning the components in $\largecomp{G \setminus X}$ to the left or the right of the bucket~$4$. Let $V_1$ denote the vertex set assigned to the buckets 1 to 3.
Buckets $1$ to $3$ have a joint capacity of $3\ell$ (as $\CCC$ is left-packed), and there are at most $\binom{3\ell}{\ell}$ choices to assign a subset of $V_1$ of size $\CCC'[1]$ to bucket~$1$.
Buckets $5$ to $8$ have a joint capacity of $n-4\ell$, and there are at most $\binom{n-4\ell}{\ell}$ choices to assign a subset of $V\setminus (V_1\cup X)$ of size $\CCC'[7]$ to bucket~$7$. The worst-case running time for this case is $$\max_{\, \ell \in I_8} \cO^* \left (\binom{n}{\ell} \cdot 2^{\cO(\sqrt{n})} \cdot \max\left \{\binom{3\ell}{\ell}, \binom{n-4\ell}{\ell} \right \} \right ).$$



\paragraph{$\mathbf{k=9}$ or $\mathbf{k=11}$.} The algorithm goes through all subsets $X \subseteq V$ with $|X|=\ell$. $X$ is assigned to the bucket~$\lceil k/2 \rceil$. As in the previous two cases, Proposition~\ref{prop:leftright} is invoked for $G[\smallcompvert{G \setminus X} \cup X]$ (with $X$ assigned to the bucket~$\lceil k/2 \rceil$). For each capacity vector generated by Proposition~\ref{prop:leftright}, the algorithm considers every possible way of assigning each connected component in $\largecomp{G\setminus X}$ to the buckets $1$ to $\lceil k/2 \rceil - 1$ or to the buckets $\lceil k/2 \rceil +1$ to $k$. Each assignment gives rise to two independent subproblems --- one on vertices $V_1$ assigned to the buckets $1$ to $(k-1)/2$, and one on vertices $V_2$ assigned to the buckets $(k+3)/2$ to $k$ (with vertices in $V_1$ and $V_2$ neighboring the vertices in $X$ assigned to the buckets $(k-1)/2$ and $(k+3)/2$, respectively). The algorithm solves these subproblems recursively as in the cases with $4$ or $5$ buckets. Let $n_1 = \lceil (n-\ell)/2 \rceil$. Then, the worst-case running times are $\max_{\, \ell \in I_k} \cO^*(\binom{n}{\ell} \cdot 2^{\cO(\sqrt{n})} \cdot \binom{n_1}{\ell})$.

\paragraph{$\mathbf{k=10}$ or $\mathbf{k=12}$.}  The algorithm uses a left-packed $(n,\ell)$-capacity vector $\CCC$ for these cases. The algorithm goes through all subsets $X \subseteq V$ with $|X|=\ell$. $X$ is assigned to the bucket~$k/2$.
The remaining analysis is similar to the previous cases.
In the two independent subproblems, generated for each capacity vector $\CCC'$ and assignment of the connected components in $\largecomp{G\setminus X}$ to the left or right of bucket~$k/2$,
the algorithm fills the buckets $3$ and $8$ if $k=10$, and the bucket 3 and the union of the buckets 9 and 10 if $k=12$.
For $k=10$, the worst-case running time is
$\max_{\, \ell \in I_{10}} \cO^*(\binom{n}{\ell} \cdot 2^{\cO(\sqrt{n})} \cdot \max \{ \binom{4\ell}{\ell}, \binom{n-5\ell}{\ell} \})$.
For $k=12$, the worst-case running time is $\max_{\, \ell \in I_{12}} \cO^*(\binom{n}{\ell} \cdot 2^{\cO(\sqrt{n})} \cdot \max \{ \binom{5\ell}{\ell}, \binom{n-6\ell}{2\ell} \})$.

\paragraph{$\mathbf{13 \le k \le 26}$.} The algorithm enumerates all subsets $X \subseteq V$ with $|X|=\ell$. $X$ is assigned to the bucket~$\lceil k/2 \rceil$. 
As in the previous cases, Proposition~\ref{prop:leftright} is invoked for $G[\smallcompvert{G \setminus X} \cup X]$. For each capacity vector generated by Proposition~\ref{prop:leftright}, the algorithm looks at every possible way of assigning each connected component in $\largecomp{G\setminus X}$ to the buckets $1$ to $\lceil k/2 \rceil-1$ or to the buckets $\lceil k/2 \rceil +1$ to $k$. Each assignment gives rise to two independent subproblems. 
For each of these two subproblems, the algorithm proceeds recursively until reaching subproblems with at most $2$ consecutive empty buckets, which can be solved by Proposition~\ref{prop:dp} in polynomial time. If $k\le 23$, this recursion has depth $3$, giving a running time of
\begin{align*}
\max_{\, \ell \in I_{k}} \, \cO^*\left(\binom{n}{\ell} \cdot 2^{\cO(\sqrt{n})} \cdot \binom{n/2}{\ell} \cdot 2^{\cO(\sqrt{n})} \cdot \binom{n/4}{\ell} \cdot 2^{\cO(\sqrt{n})} \right)\enspace.
\end{align*}
If $24 \le k \le 26$, the recursion has depth $4$, giving a running time of
\begin{align*}
\max_{\, \ell \in I_{k}} \, \cO^*\left(\binom{n}{\ell} \cdot 2^{\cO(\sqrt{n})} \cdot \binom{n/2}{\ell} \cdot 2^{\cO(\sqrt{n})} \cdot \binom{n/4}{\ell} \cdot 2^{\cO(\sqrt{n})} \cdot \binom{n/8}{\ell} \cdot 2^{\cO(\sqrt{n})} \right) \enspace.
\end{align*}
 
\paragraph{$\mathbf{k \ge 27}$.} By Proposition~\ref{prop:generalalgo} 
the running time of the algorithm is bounded in this case by
\begin{align*}
\max_{\, \ell \in I_{k}} \, \cO^*\left( {n \choose \ell} \cdot {n/2 \choose \ell} \cdot 2^{4\ell} \cdot 3^{n/4} \right) \enspace.
\end{align*}

\paragraph{Main Result.} Putting together all the above arguments and using 
the
numerical values
 from Table~\ref{tab:runtimes}
we obtain our main result (Theorem~\ref{thm:2approx}). The running time is dominated by the cases where $k=7$ and $k=8$. The algorithm outputs $2\ell-1$, where $\ell$ is the smallest integer such that $G$ has a bucket arrangement with an $(n,\ell)$-capacity vector. The algorithm requires only polynomial space.

If $G$ is disconnected, the algorithm finds for each connected component $G_i=(V_i,E_i)$ the smallest $\ell_i$ such that $G_i$ has a bucket arrangement corresponding to a $(|V_i|,\ell_i)$-capacity vector and outputs $2\ell_m-1$, where $\ell_m=\max_i\{\ell_i\}$.  

\begin{theorem}[Main Theorem]\label{thm:2approx}
There is a polynomial space $2$-approximation algorithm for the \BANDW problem that takes $\cO(1.9797^n)$ time on graphs with $n$ vertices.
\end{theorem}

\section{Conclusion}
\label{sec:con}

For finding exact solutions, it is known that many problems (by subexponential time preserving reductions) do not admit subexponential time algorithms under the Exponential Time Hypothesis~\cite{ImpagliazzoP01}. The Exponential Time Hypothesis postulates that there is a constant $c>0$ such that 3-\textsc{Sat} cannot be solved in time $\cO(2^{c n})$, where $n$ is the number of variables of the input formula.  We conjecture that the \BANDW problem has no subexponential time $2$-approximation algorithm, unless the Exponential Time Hypothesis fails.

%



\paragraph{Acknowledgment} We thank the reviewers for very helpful comments, in particular to avoid overcounting in the proof of Proposition \ref{prop:dp}.
Martin F\"urer acknowledges partial support from the National Science Foundation (CCF-0728921 and CCF-0964655).
Serge Gaspers acknowledges partial support from the Norwegian Research Council and the French National Research Agency (GRAAL project ANR-06-BLAN-0148).

\bibliographystyle{plain}

\begin{thebibliography}{99}
\bibitem{AminiFS09} O. Amini, F.~V. Fomin, and S. Saurabh,
\newblock Counting subgraphs via homomorphisms,
\newblock Proceedings of ICALP 2009, 71--82 (2009).

\bibitem{AssmannPSZ81} S.~F. Assmann, G.~W. Peck, M.~M. Sys{\l}o, and J. Zak,
\newblock The bandwidth of caterpillars with hairs of length 1 and 2,
\newblock {\em SIAM J. Algebra. Discr.} 2, 387--393 (1981).

\bibitem{BjorklundH06} A. Bj{\"o}rklund, T. Husfeldt, and M. Koivisto,
\newblock Set partitioning via inclusion-exclusion,
\newblock {\em SIAM J. Comput.} 39(2), 546--563 (2009).

\bibitem{BlumKRV00} A. Blum, G. Konjevod, R. Ravi, and S. Vempala,
\newblock Semi-definite relaxations for minimum bandwidth and other vertex-ordering problems,
\newblock {\em Theor. Comput. Sci.} 235(1), 25--42 (2000).

\bibitem{BodlaenderFH94} H.~L. Bodlaender, M.~R. Fellows, and M.~T. Hallett,
\newblock Beyond NP-completeness for problems of bounded width: hardness for the W-hierarchy,
\newblock Proceedings of STOC 1994, 449--458 (1994).

\bibitem{BourgeoisEP08} N. Bourgeois, B. Escoffier, and V.~Th. Paschos,
\newblock Approximation of max independent set, min vertex cover and related problems by moderately exponential algorithms,
\newblock {\em Discrete Appl. Math.} 159(17), 1954--1970 (2011).

\bibitem{BourgeoisEP08coloring} N. Bourgeois, B. Escoffier, and V.~Th. Paschos,
\newblock Approximation of min coloring by moderately exponential algorithms,
\newblock {\em Inform. Process. Lett.} 109(16), 950--954 (2009).

\bibitem{BourgeoisEP09} N. Bourgeois, B. Escoffier, and V.~Th. Paschos,
\newblock Efficient approximation of min set cover by moderately exponential algorithms,
\newblock {\em Theor. Comput. Sci.} 410(21-23), 2184--2195 (2009).

\bibitem{BourgeoisLMP10} N. Bourgeois, G. Lucarelli, I. Milis, and V.~Th. Paschos,
\newblock Approximating the max-edge-coloring problem,
\newblock {\em Theor. Comput. Sci.} 411(34-36), 3055--3067 (2010).

\bibitem{CaiH06} L. Cai and X. Huang,
\newblock Fixed-parameter approximation: conceptual framework and approximability results,
\newblock {\em Algorithmica} 57(2), 398--412 (2010).

\bibitem{ChenGG06} Y. Chen, M. Grohe, and M. Gr{\"u}ber,
\newblock On parameterized approximability,
\newblock Proceedings of IWPEC 2006, 109--120 (2006).

\bibitem{ChenHKX04} J. Chen, X. Huang, I.~A. Kanj, and G. Xia,
\newblock Linear FPT reductions and computational lower bounds,
\newblock Proceedings of STOC 2004, 212--221 (2004).

\bibitem{CyganKPW08} M. Cygan, L. Kowalik, and M. Wykurz,
\newblock Exponential-time approximation of weighted set cover,
\newblock {\em Inf. Process. Lett.} 109(16), 957--96 (2009).

\bibitem{CyganP10} M. Cygan and M. Pilipczuk,
\newblock Bandwidth and distortion revisited,
\newblock {\em Discrete Appl. Math.} 160(4-5), 494--504 (2012).

\bibitem{CyganP09} M. Cygan and M. Pilipczuk,
\newblock Even faster exact bandwidth,
\newblock {\em ACM T. Algorithms} 8(1), 8:1--8:14 (2012).

\bibitem{CyganP09icalp} M. Cygan and M. Pilipczuk,
\newblock Exact and approximate bandwidth,
\newblock {\em Theoret. Comput. Sci.} 411(40-42), 3701--3713 (2010).

\bibitem{DantsinGHK01} E. Dantsin, M. Gavrilovich, E.~A. Hirsch, and B. Konev,
\newblock MAX SAT approximation beyond the limits of polynomial-time approximation,
\newblock {\em Ann. Pure and Appl. Logic} 113(1-3), 81--94 (2001).

\bibitem{DowneyFM06} R.~G. Downey, M.~R. Fellows, and C. McCartin,
\newblock Parameterized approximation problems,
\newblock Proceedings of IWPEC 2006, 121--129 (2006).

\bibitem{DunaganV01} J. Dunagan and S. Vempala,
\newblock On euclidean embeddings and bandwidth minimization,
\newblock Proceedings of RANDOM-APPROX 2001, 229--240 (2001).

\bibitem{DyerFKKPV93} M.~E. Dyer, A.~M. Frieze, R. Kannan, A. Kapoor, L. Perkovic, and U.~V. Vazirani,
\newblock A mildly exponential time algorithm for approximating the number of solutions to a multidimensional knapsack problem,
\newblock {\em Combin. Probab. Comput.} 2, 271--284 (1993).

\bibitem{Feige00approx} U. Feige,
\newblock Approximating the bandwidth via volume respecting embeddings,
\newblock {\em J. Comput. Syst. Sci.} 60(3), 510--539 (2000).

\bibitem{Feige00survey} U. Feige,
\newblock Coping with the NP-hardness of the graph bandwidth problem,
\newblock Proceedings of SWAT 2000, 10--19 (2000).

\bibitem{FeigeT05} U. Feige and K. Talwar,
\newblock Approximating the bandwidth of caterpillars,
\newblock {\em Algorithmica} 55(1), 190--204 (2009).

\bibitem{FurerGK09} M. F{\"u}rer, S. Gaspers, and S.~P. Kasiviswanathan,
\newblock An exponential time 2-approximation algorithm for bandwidth,
\newblock Proceedings of IWPEC 2009, 173--184 (2009).

\bibitem{GareyGJK78} M.~R. Garey, R.~L. Graham, D.~S. Johnson, and D.~E. Knuth,
\newblock Complexity results for bandwidth minimization,
\newblock {\em SIAM J. Appl. Math.} 34(3), 477--495 (1978).

\bibitem{GurariS84} E.~M. Gurari and I.~H. Sudborough,
\newblock Improved dynamic programming algorithms for bandwidth minimization and the MinCut Linear Arrangement problem,
\newblock {\em J. Algorithms} 5(4), 531--546, 1984.

\bibitem{HeggernesKM09} P. Heggernes, D. Kratsch, and D. Meister,
\newblock Bandwidth of bipartite permutation graphs in polynomial time,
\newblock {\em Journal of Discrete Algorithms} 7(4), 533--544 (2009).

\bibitem{Hirsch03} E.~A. Hirsch,
\newblock Worst-case study of local search for Max-$k$-SAT,
\newblock {\em Discrete Appl. Math.} 130(2), 173--184 (2003).

\bibitem{ImpagliazzoP01} R. Impagliazzo and R. Paturi,
\newblock On the complexity of k-SAT,
\newblock {\em J. Comput. Syst. Sci.} 62(2), 367--375 (2001).


\bibitem{JerrumV96} M. Jerrum and U.~V. Vazirani,
\newblock A mildly exponential approximation algorithm for the permanent.
\newblock {\em Algorithmica} 16(4-5), 392--401 (1996).

\bibitem{jrlee} J.~R. Lee,
\newblock Volume distortion for subsets of euclidean spaces,
\newblock {\em Discrete Comput. Geom.} 41(4), 590--615 (2009).

\bibitem{KleitmanV90} D.~J. Kleitman and R.~V. Vohra,
\newblock Computing the bandwidth of interval graphs,
\newblock {\em SIAM J. Discrete Math.} 3, 373--375 (1990).

\bibitem{KloksKM99} T. Kloks, D. Kratsch, and H. M{\"u}ller,
\newblock Approximating the bandwidth for asteroidal triple-free graphs.
\newblock {\em J. Algorithms} 32(1), 41--57 (1999).

\bibitem{KloksKM98} T. Kloks, D. Kratsch, and H. M{\"u}ller,
\newblock Bandwidth of chain graphs.
\newblock {\em Inform. Process. Lett.} 68, 313--315 (1998).

\bibitem{MaheshRS91} R. Mahesh, C.~P. Rangan, and A. Srinivasan,
\newblock On finding the minimum bandwidth of interval graphs,
\newblock {\em Inform. Comput.} 95(2), 218--224 (1991).

\bibitem{Marx08} D. Marx,
\newblock Parameterized complexity and approximation algorithms,
\newblock {\em Comput. J.} 51(1), 60--78 (2008).

\bibitem{Monien86} B. Monien,
\newblock The bandwidth minimization problem for caterpillars with hair length 3 is NP-complete,
\newblock {\em SIAM J. Algebra. Discr.} 7(4), 505--512 (1986).

\bibitem{MonienS80} B. Monien and I.~H. Sudborough,
\newblock Bandwidth problems in graphs,
\newblock Proceedings of Allerton Conference on Communication, Control, and Computing 1980, 650--659 (1980).

\bibitem{Papadimitriou76} C. Papadimitriou,
\newblock The NP-completeness of the bandwidth minimization problem,
\newblock {\em Computing} 16, 263--270 (1976).

\bibitem{RamanSS07} V. Raman, S. Saurabh, and S. Sikdar,
\newblock Efficient exact algorithms through enumerating maximal independent sets and other techniques,
\newblock {\em Theor. Comput. Syst.}, 41(3), 563--587 (2007).


\bibitem{ShresthaTU11} A.~M.~S.~Shrestha, S. Tayu, and S. Ueno,
\newblock Bandwidth of convex bipartite graphs and related graph classes,
\newblock Proceedings of COCOON 2011, 307--318 (2011).


\bibitem{Sprague94} A.~P. Sprague,
\newblock An $O(n \log n)$ algorithm for bandwidth of interval graphs.
\newblock {\em SIAM J. Discrete Math.} 7, 213--220 (1994).

\bibitem{Unger98} W. Unger,
\newblock The complexity of the approximation of the bandwidth problem,
\newblock Proceedings of FOCS 1998, 82--91 (1998).

\bibitem{VassilevskaWW06} V. Vassilevska, R. Williams, and S.~L.~M. Woo,
\newblock Confronting hardness using a hybrid approach,
\newblock Proceedings of SODA 2006, 1--10 (2006).

\bibitem{Yan97} J.~H. Yan,
\newblock The bandwidth problem in cographs,
\newblock {\em Tamsui Oxf. J. Math. Sci.} 13, 31--36 (1997).
\end{thebibliography}
\begin{small}

\end{small}
\end{document}